%% file: tamc_56.tex
\documentclass[runningheads]{llncs}
\usepackage{epsfig,graphicx,subfigure}
\usepackage{mathtools}
\usepackage{hyperref}
\usepackage{xcolor} % For textcolor

\newcommand{\etal}{\textit{et al}.}

\begin{document}
%
%\title{Analysis of Hidden Community Detection \\ on Multi-layer Stochastic Models}
%\title{Theoretical Prospective of Hidden Community Detection}
\title{Hidden Community Detection on Two-layer Stochastic Models: a Theoretical Perspective}
\titlerunning{Hidden Community Detection on Two-layer Stochastic Models}
% If the paper title is too long for the running head, you can set
% an abbreviated paper title here
%
\author{Jialu Bao\inst{1} \thanks{Portion of the work was done while at Cornell University} \and
Kun He \inst{2} \thanks{Corresponding author. Email: brooklet60@hust.edu.cn} \and
Xiaodong Xin \inst{2} \and
Bart Selman\inst{3} \and 
John E. Hopcroft\inst{3}}
\authorrunning{J. Bao, K. He, et al.}
% First names are abbreviated in the running head.
% If there are more than two authors, 'et al.' is used.
%
\institute{Department of Computer Science, University of Wisconsin-Madison,\\ Madison, WI 50706, USA %\email{jialu@cs.wisc.edu}
\and
School of Computer Science and Technology, Huazhong University of Science and Technology, Wuhan 430074, China
%\email{brooklet60@hust.edu.cn, m201973348@hust.edu.cn}
\and
Department of Computer Science, Cornell University, Ithaca, NY 14853
%\email{selman@cs.cornell.edu, jeh@cs.cornell.edu}
}
\maketitle              % typeset the header of the contribution
\begin{abstract}
%(The abstract should briefly summarize the contents of the paper in
%15--250 words.——)
Hidden community is a new graph-theoretical concept recently proposed by~\cite{he18}, in which the authors also propose a meta-approach called HICODE (Hidden Community Detection) for detecting hidden communities. HICODE is demonstrated through experiments that it is able to uncover previously overshadowed weak layers and uncover both weak and strong layers at a higher accuracy. However, the authors provide no theoretical guarantee for the performance. In this work, we focus on theoretical analysis of HICODE on synthetic two-layer networks, where layers are independent of each other and each layer is generated by stochastic block model.  We bridge their gap through two-layer stochastic block model networks in the following aspects: 1) we show that partitions that locally optimize modularity correspond to grounded layers, indicating modularity-optimizing algorithms can detect strong layers; 2) we prove that when reducing found layers, HICODE increases absolute modularities of all unreduced layers, showing its layer reduction step makes weak layers more detectable. Our work builds a solid theoretical base for HICODE, demonstrating that it is promising in uncovering both weak and strong layers of communities in two-layer networks.
% 180 words

\end{abstract}
%
%\section{Introduction}
\input{intro.tex}

\section{Preliminary}
%Explain HICODE, reduceEdge, reduceWeight, removeEdge, multi-layer stochastic blockmodel.
In this section, we first introduce metrics that measure community partition quality. Then, we summarize important components in HICODE, the iterative meta-approach we are going to analyze, and in particular, how it reduce layers of detected communities during the iterations. Also, we define the multi-layer stochastic block model formally, and the rationale why it is a reasonable abstraction of generative processes of real world networks. 

\subsection{Modularity metric}
In determining plausible underlying communities in a network, we rely on metrics measuring quality of community partitions. Usually, nodes sharing common communities are more likely to develop connections with each other, so in single-layer networks, we expect that most edges are internal to one grounded community, instead of outgoing edges whose two endpoints belong to two communities. It thus gives rise to metrics measuring the similarity between an arbitrary partition and the grounded partition based on the fraction of internal edges to outgoing edges. One widely-used metric of this kind is ``modularity''~\cite{girvan2002community}. 
We define the modularity of one community in multi-layer networks as follows:

\begin{definition}[Modularity of a community] 
	Given a graph $G = (V, E)$ with a total of $e$ edges and multiple layers of communities, where each layer of communities partitions all nodes in the graph, for a community $i$ in layer $l$, let $e_{ll}^i$ denote $i$'s internal edges, and $e_{lout}^i$ denote the number of edges that have exactly one endpoint in community $i$. Let $d_l^i$ be the total degree of nodes in community $i$ ($d_l^i = 2 e_{ll}^{i} + e_{lout}^i$). Then the modularity of community $i$ in layer $l$ is $Q_l^i = \frac{e_{ll}^{i}}{e} - \left(\frac{d_{ll}^{i}}{2e} \right)^2$.
\end{definition}

Roughly, the higher fraction of internal edges a community has among all edges, the higher its modularity in graph, indicating that members in that community are more closely connected. 
% The value of modularity lies in the range $[-\frac{1}{4}, 1)$. KunHe{not sure for one community}

When optimizing modularity, the algorithm concerns the modularity of a partition instead of one community. The modularity of a partition is defined as follows, which is consistent with the original definition of Girvan \etal ~\cite{girvan2002community}: 

\begin{definition}[Modularity of a partition/layer]
Given a network $G = (V,E)$ with multiple layers of communities, for any layer $l$, say $l$ partitions all the nodes into disjoint communities 
$\{1, \dots, N\}$, then the layer modularity is
$Q_l = \sum_{i=1}^N Q_l^i$.
\end{definition}
Whether in single-layer network or multi-layer ones, the ground truth community partition is expected to have high modularity when compared to other possible partitions. 

\subsection{HIdden COmmunity DEtection (HICODE) algorithm}
Informally, given a state-of-the-art community detection algorithm $\mathcal{A}$ for single layer networks, HICODE$(\mathcal{A})$ finds all layers in multi-layer networks through careful alternations of detecting the strongest layer in the remaining graph using $\mathcal{A}$ and reducing found layers on the network. 
Given a network $G = (V,E)$, He \etal ~\cite{he18} proposed three slightly different methods for reducing layers in HICODE:
\begin{enumerate}
\item \textbf{RemoveEdge: } Given one layer $l$ that partitions $G$, \textit{RemoveEdge} removes all internal edges of layer $l$ from $G$. 
\item \textbf{ReduceEdge: } Given one layer $l$ that partitions $G$, \textit{ReduceEdge} approximates the background density $q$ of edges contributed by all other layers, and then removes $1-q$ fraction of internal edges of layer $l$ from network $G$. We will detail the computation of $q$ after introducing multi-layer stochastic block model. 
\item \textbf{ReduceWeight: } This is the counterpart of ReduceEdge on weighted graphs. Given one layer $l$ that partitions network $G$, \textit{ReduceWeight} approximates the background density $q$ of edges contributed by all other layers, and then reduces the weight of all internal edges to a $q$ fraction of its original values. 
\end{enumerate}

For detailed description of HICODE, see Appendix A. 

\subsection{Multi-layer stochastic block model}
Before defining the general multi-layer Stochastic Block Model (SBM), consider the case where there is exactly two layers.

\begin{definition} [Two-layer Stochastic Block Model] 
A synthetic network $G(n, n_{1}, p_{1}, n_2, p_2 )$ generated by two-layer stochastic block model has $n$ nodes, where $n, n_{1}, n_{2} \in N^{+}, n_{1}, n_{2} \geq 3$. For $l = 1$ or $2$, layer $l$ of $G$ consists of $n_l$ planted communities of size $s_{l} = \frac{n}{n_{l}}$ with internal edge probability $p_{l} \in (0,1]$. Communities in different layers are grouped independently, so they are expected to intersect with each other by $r = \frac{n}{n_{1}n_{2}}$ nodes.
\end{definition}

Each community of layer $l$ is expected to have $p_l \cdot \frac{1}{2} s_1^l$ internal edges~\footnote{For simplicity, we allow self-loops.}. The model represents an ideal scenario when there is no noise and all outgoing edges of one layer are the result of them being internal edges of some other layers. We will detail the expected number of outgoing edges and the size of the intersection block of layers in Lemma 1 in the next section.

For example, in $G(200, 4, 5, p_1, p_2)$, layer 1 contains four communities $C_{1}^{1} = \{1,2,..., 50\}$, $C_{1}^{2} = \{51,52,..., 100\}$, $C_{1}^{3} = \{101,102,..., 150\}$, $C_{1}^{4} = \{151,152,...,$ $ 200\}$, and layer 2 contains five communities $C_{2}^{1} = \{1,6, ..., 196\}$, $C_{2}^{2} = \{2,7, ...,$ $197\}$,  $C_{2}^{3} = \{3,8, ..., 198\}$,  $C_{2}^{4} = \{4,9, ..., 199\}$, $C_{2}^{5} = \{5,10, ..., 200\}$. 
Each community is modeled as an Erd\H{o}s-R\'{e}nyi graph. %p_1=0.12, p_2=0.10
Each $C_{1}^{i}$ in layer 1 is expected to have $0.5 \cdot 50^2 p_1$ internal edges, and each  $C_{2}^{i}$ in layer 2 are expected to have $0.5 \cdot 40^2 p_1$ internal edges.

Each community in layer 1 overlaps with each community in layer 2. Each overlap consists of 20\% of the nodes of layer 1 community and 25\% of the nodes of layer 2 community.
Fig. \ref{fig:SynAnalysis} (a) and (b) show the adjacency matrix when nodes are ordered by $[1,..., n]$ for layer 1, and $[1, 6, ..., 196,$ $2,7, ..., 197, 5,10, ..., 200]$ for layer 2, respectively (Here we set $p_1 = 0.12, p_2 = 0.10$).
Fig. \ref{fig:SynAnalysis} (c) and (d) show an enlarged block for each layer.
Edges in layer 1 are plotted in red, edges in layer 2 are plotted in blue and the intersected edges are plotted in green.

\vspace{-1.5em}
\begin{figure}[htbp!] 
  \subfigure[layer 1]{
    \begin{minipage}[b]{0.23\textwidth}
      \centering
      \includegraphics[width=0.9in]{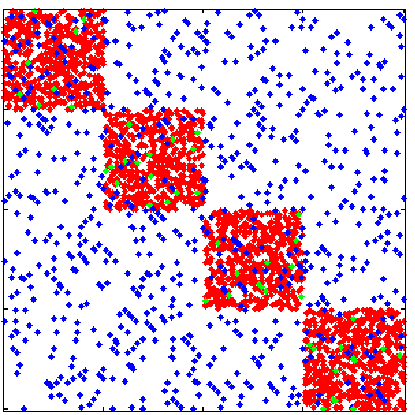}
    \end{minipage}}
  \subfigure[layer 2]{
    \begin{minipage}[b]{0.23\textwidth}
      \centering
      \includegraphics[width=0.9in]{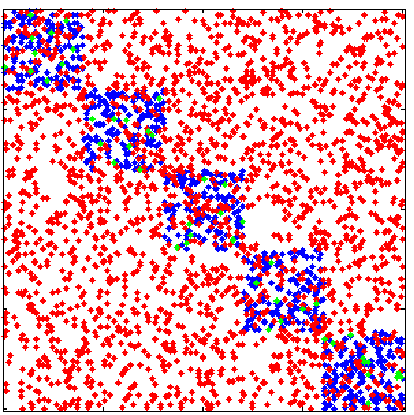}
    \end{minipage}}
  \subfigure[\small{a $L_1$ block}]{
    \begin{minipage}[b]{0.24\textwidth}
      \centering
      \includegraphics[width=0.9in]{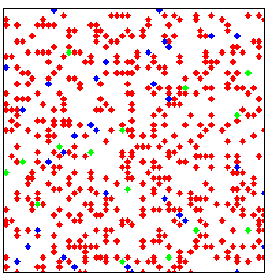}
    \end{minipage}}
  \subfigure[\small{a $L_2$ block}]{
    \begin{minipage}[b]{0.24\textwidth}
      \centering
      \includegraphics[width=0.9in]{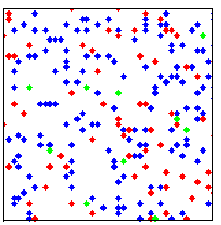}
    \end{minipage}}
  \vspace{-1.5em}
  \caption{The stochastic blocks in two layers.}
  \label{fig:SynAnalysis} %% label for entire figure
\vspace{-1.5em}  
\end{figure}

More generally, we can define a multi-layer stochastic block model.

\begin{definition} [Multi-layer Stochastic Block Model]
A multi-layer stochastic block model $G(n, n_{1}, p_{1} , ..., n_{L}, p_{L})$ generates a network with $L$ layers, and each layer $l$ has $n_l$ communities of size $\frac{n}{n_l}$ with internal edge probability $p_l$. All layers are independent with each other. 
\end{definition}

\subsection{Background edge probability for multi-layer SBM}

Given a layer, observed edge probability within its grounded communities would be higher than its grounded edge generating probability, because other layers could also generate edge internal to this layer. When we are interested in the grounded edge generating probability of a layer, we can consider edges generated by all other layers as background noise. Since layers are independent to each other, these background noise edges are uniformly distributed among communities of layer $l$, so we can expect background noise edge probability the same on node pairs either internal to or across layer $l$ communities. Thus, the observed edge probability $\widehat{p}$ of communities in a layer $l$ equals $ p + \widehat{q} - p \cdot \widehat{q}$, where $p$ is the grounded edge generating probability of layer $l$, $\widehat{q}$ is the observed edge probability across layer 1 communities. Thus, we can estimate the actual edge probability by $p = \frac{\widehat{p} - \widehat{q}}{1-\widehat{q}}$.

\section{Theoretical analysis on two-layer SBM}

In this section, we show that on networks generated by two-layer stochastic block model, weakening one layer would not decrease the quality of communities in any other layer even when they considerably overlap with each other. We will prove on two-layer stochastic block models that absolute modularity of unreduced layer must increase after performing RemoveEdge, ReduceEdge, or ReduceWeight.
%, and then generalize the analysis to networks with more than two layers (see analysis for multiple layers in Appendix). 
For simplicity, we make the assumption that the base algorithm can uncover a layer exactly -- every time it finds a layer to reduce, it does not make mistakes on community membership. This is a strong assumption, but later on we will justify why our result still holds if the base algorithm only approximates layers and why the base algorithm can almost always find some approximate layers.

For each community in layer $l$, let $s_l$ denote the size of each community in layer $l$, and  $m_l$ denote the number of node pairs in the community. Since we allow self-loops, $m_l = \frac{1}{2} s_l^2$. Also, with the assumption that all communities in one layer are equal sized, their expected numbers of internal (or outgoing) edges are the same. Thus, we can use $e_{ll}$, $e_{lout}$ to respectively denote the expected number of internal, outgoing edges for each community $i$ in layer $l$. Then, let $d_l = 2e_{ll} + e_{lout}$ denote the expected total degree of any community in layer $l$. 

\begin{lemma}
In the synthetic two-layer block model network $G(n, n_{1}, n_{2}, p_{1}, p_{2})$, for a given community $i$ in layer 1, the expected number of its internal edges as well as outgoing edges, and layer 1's modularity are as follows:
\begin{align}
     e_{11} &=  \left(1-\frac{1}{n_{2}}\right)m_{1}p_{1} + \frac{1}{n_{2}}m_{1}p_{12},\\
    e_{1out} &=  \frac{p_{2}}{n_{2}} s_{1}(n - s_{1}),\\
    Q_1 &= 1 - \frac{1}{n_1} - \frac{e_{1out}}{d_1},
\end{align}
where $p_{12} = p_1 + p_2 - p_1 \cdot p_2$. Symmetrically, given a community $i$ in layer 2, the expected number of its internal edges as well as outgoing edges, and layer 2's modularity are as follows: 
\begin{align}
    e_{22} &= \left(1-\frac{1}{n_{1}}\right)m_{2}p_{2} + \frac{1}{n_{1}}m_{2}p_{12},\\
    e_{2out} &=  \frac{p_1}{n_1} s_{2}(n - s_{2}),\\
    Q_2 &= 1 - \frac{1}{n_2} - \frac{e_{2out}}{d_2}.
\end{align}
\end{lemma}

For detailed proofs, see Appendix B.  %\ref{appendix2layerproof}. 

\begin{lemma}
For layer $l$ in a two-layer stochastic block model, if the layer weakening method (e.g. RemoveEdge, ReduceEdge, ReduceWeight) reduces a bigger percentage of outgoing edges than internal edges, i.e. the expected number of internal and outgoing edges after weakening, $e'_{ll}, e'_{l out}$, satisfy $\frac{e'_{l out}}{e_{l out}} < \frac{e'_{l l}}{e_{l l}}$, then the modularity of layer $l$ increases after the weakening method.
\end{lemma}

For detailed proofs, see Appendix B. 

%\begin{definition}
For a synthetic stochastic block model network $G$ with set of layers $\mathcal{L}$, let $S_l$  be the set of edges whose underlying node pairs are only internal to layer $l \subseteq \mathcal{L}$, let $S_{l_1l_2}$ be the set of edges internal to both layers $l_1, l_2 \subseteq \mathcal{L}$.
%\end{definition}
Concretely, in the two-layer stochastic block model, $\mathcal{L} = \{ 1, 2\}$. $S_1$ is the set of edges only internal to layer 1, $S_2$ is the set of edges only internal to layer 2, and $S_{12}$ is the set of edges internal to both layer 1 and layer 2.  

\begin{lemma}
In a two-layer stochastic blockmodel network $G(n, n_1, n_2, p_1, p_2)$, before any weakening procedure. 
\begin{align*}
&e_{11} = \frac{|S_{12}| + |S_1|}{ n_1}, &e_{1out} = \frac{2}{n_1} |S_2|,\\
&e_{22} = \frac{|S_{12}| + |S_2|}{ n_2}, &e_{2out} = \frac{2}{n_2} |S_1|.
\end{align*}
\end{lemma}

For detailed proofs, see Appendix B.

Using the above three lemmas, we can prove the following theorems.

\begin{theorem}
\textit{For a two-layer stochastic blockmodel network $G(n,n_1,n_2,p_1,p_2)$, the modularity of a layer increases if we apply \textcolor{blue}{RemoveEdge} on communities in the other layer.}
\end{theorem}

\begin{proof}
If we remove all internal edges of communities in layer $1$,  both $|S_{12}|$ and $|S_1|$ become 0, then the remaining internal edges of layer $2$ is $e'_{22} = \frac{1}{n_2}(|S_{12}|+|S_2|)=\frac{|S_2|}{n_2}> 0$. There is no outgoing edge of layer 2, so $e'_{2out} = 0$. Thus, $\frac{e'_{2 out}}{e_{2 out}}=0 < \frac{e'_{2 2}}{e_{2 2}}$, and applying Lemma 2, we have that the modularity of layer 2 after RemoveEdge on layer 1 $Q'_2 >  Q_2$. 

Similarly, the modularity of layer 1 after RemoveEdge on layer 2, $Q'_1 $, is greater than $Q_1$.   
\end{proof}

 \textcolor{blue}{RemoveEdge} not only guarantees to increase the absolute modularity of layer 2 but also guarantees that layer 2 would have higher modularity than any possible partition of $n$ nodes into $n_2$ communities in the reduced network. 
\begin{theorem}
For a two-layer stochastic blockmodel network $G(n,n_1,n_2,p_1,p_2)$, If no layer 2 community contains more than half of the total edges inside it after applying \textcolor{blue}{RemoveEdge} on layer 1, then layer 2 has the highest modularity among all possible partitions of $n$ nodes into $n_2$ communities.
\end{theorem}
\begin{proof}
After applying \textcolor{blue}{RemoveEdge} on layer 1, there are no outgoing edges of any community in layer 2. It means that for any community $i$, $e_{2out}^{i}$=0 and $d^{i}_2=2e_{22}^{i}$. Thus, the modularity of layer 2 is:
\begin{align*}
	Q_2 =  &\sum_{i \in \text{layer 2}} Q_2^i 
    = \sum_{i \in \text{layer 2}}\left[\frac{e_{22}^{i}}{e}-\left(\frac{d_2^{i}}{2e}\right)^2 \right] \\
    = &\sum_{i \in \text{layer 2}}\left[\frac{4e \cdot e^{i}_{22} - (2e^{i}_{22})^2}{4e^2}\right] 
    = n_2\left(\frac{e \cdot e_{22} - (e_{22})^2}{e^2}\right). 
\end{align*}
For any one partition, we can transform layer 2 partition to it by moving a series of nodes across communities. Every time we move one node from one community $i$ to another community $j$, both $e_{2out}^{i}, e_{2out}^{j}$ will increase by 1, $e_{22}^{i}$ will decrease by 2 while $e_{22}^{j}$ remains the same. Let $e'^{i}_{2out}, e'^{i}_{22}$ denote corresponding values after all movements. The following always holds no matter how many times we move:
\begin{align*}
	2\sum_{i \in \text{layer 2}} (e^{i}_{22}-e'^{i}_{22})&=\sum_{i \in \text{layer 2}}e'^{i}_{2out}
%	2\sum_{i \in \text{layer 2}} e^{i}_{22} &= 2 \sum_{i \in \text{layer 2}} e'^i_{22} + \sum_{i \in \text{layer 2}} e'^i_{2out}  
\end{align*}
Now $Q'_2$, the modularity of the new partition after moving, is:
\begin{align*}
    Q'_2= &\sum_{i \in \text{layer 2}} \frac{e'^{i}_{22}}{e}-\left(\frac{d'^{i}_2}{2e}\right)^2 \\
    = &\sum_{i \in \text{layer 2}}\frac{4e\cdot e'^{i}_{22}}{4e^2}-\sum_{i \in \text{layer 2}}\frac{(2e'^{i}_{22}+e'^{i}_{2out})^{2}}{4e^2}.
\end{align*}
Let $e_{22}^{i}-e'^{i}_{22}=\Delta_i$. Because of $ (a+b)^2 \geq a^2 + b^2  $ for  any $a,b\geq 0$, we have:
\begin{align*}
    Q'_2 \leq&\sum_{i \in \text{layer 2}}\frac{4e\cdot e'^{i}_{22}}{4e^2}-\sum_{i \in \text{layer 2}}\frac{(2e'^{i}_{22})^{2}+(e'^{i}_{2out})^{2}}{4e^2} \\
    = &\frac{4e\cdot \sum e'^{i}_{22}-\sum 4(e'^{i}_{22})^{2}-\sum (e'^{i}_{2out})^{2}}{4e^{2}} \\
    = &\frac{4e\cdot \sum(e^{i}_{22}-\Delta_i)-\sum 4(e^{i}_{22}-\Delta_i)^{2}-\sum (e'^{i}_{2out})^{2}}{4e^2} \\
    = & Q_2+ \frac{8\sum \Delta_i e^{i}_{22}-4e\cdot\sum \Delta_i-\sum (e'^{i}_{2out})^2-4\sum \Delta_i^{2}}{4e^2}
\end{align*}
Let $T$ abbreviate $8\sum \Delta_i e^{i}_{22}-4e\cdot\sum \Delta_i-\sum (e'^{i}_{2out})^2-4\sum \Delta_i^{2}$, then $Q'_2 = Q_2 + \frac{T}{4e^2}$. When no layer 2 community contains more than half of the total edges after applying \textcolor{blue}{RemoveEdge} on layer 1, i.e.,  $e_{22}^i\leq\frac{e}{2}$, 
\begin{align*}
    T = &8\sum \Delta_i e^{i}_{22}-4e\cdot\sum \Delta_i-\sum (e'^{i}_{2out})^2-4\sum \Delta_i^{2} \\
    \leq&4e\cdot\sum \Delta_i-4e\cdot\sum \Delta_i-\sum (e'^{i}_{2out})^2-4\sum \Delta_i^{2} \leq 0.
\end{align*}
Finally, we have $Q'_2 \leq Q_2+\frac{T}{4e^2}\leq Q_2$. 
Hence, layer 2 has the highest modularity among all possible partitions of $n$ nodes into $n_2$ communities. In this way, RemoveEdge makes the unreduced layer easier for the base algorithm to detect.
\end{proof}

\begin{theorem}
For a two-layer stochastic blockmodel network $G(n,n_1,n_2,p_1,p_2)$, the modularity of a layer increases if we apply \textcolor{blue}{ReduceEdge} on all communities in the other layer.
\end{theorem}

\begin{proof}
In ReduceEdge of layer 1, we keep edges in the given community with probability $q'_{1} = \frac{1-\widehat{p}}{1-\widehat{q}}$, where $\widehat{p}$ is the observed edge probability within the detected community and $\widehat{q}$ is the observed background noise.

ReduceEdge on layer 1 would only keep $q'_{1}$ fraction of edges in $S_{12}$ and $S_1$, so after ReduceEdge,  
\begin{align*}
	e'_{22} &= \frac{1}{n_2}(|S_2| + |S_{12}| \cdot q'_{1}) 
    > \frac{1}{n_2}(|S_2| + |S_{12}| ) \cdot q'_{1} 
    = e_{22} \cdot q'_{1},\\
    e'_{2out} &=  \frac{2}{n_1}|S_1| \cdot q'_{1}
    = e_{2out} \cdot q'_{1}.
\end{align*}
Thus, $\frac{e'_{2 out}}{e_{2 out}} < \frac{e'_{2 2}}{e_{2 2}}$, and Lemma 2 indicates that $Q_2 < Q'_2$. Similarly, for the modularity of layer 1 after ReduceEdge on layer 1, $Q'_1 > Q_1$.
\end{proof}

\begin{theorem}
For a synthetic two-layer block model network $G(n,n_1,n_2,p_1,p_2)$, the modularity of a layer increases if we apply \textcolor{blue}{ReduceWeight} on all communities in the other layer.
\end{theorem}
\begin{proof}
According to~\cite{he18}, ReduceWeight on layer 1 multiplies the weight of edges in layer 1 community by $q'_{1} = 1- \frac{1-\widehat{p}}{1-\widehat{q}}$ percent.
In weighted network, the weight sum of internal edges of a community $i$ in layer 2 is $e_{22} = \frac{1}{2} \sum_{u,v \in i} w_{uv} \cdot A_{uv}$ where $w_{uv}$ is the weight of edge $(u,v)$. By construction, ReduceWeight on layer 1 reduces weight of all edges in $S_{12}$ or $S_1$, but does not change weight of edges in $S_2$. Thus,
\begin{align*}
	e'^{i}_{22} &= \frac{1}{2} \sum_{ u, v \in i,\ (u,v) \in S_{12}} w_{uv} \cdot A_{uv} \cdot q'_{1} +  \frac{1}{2} \sum_{ u, v \in i,\ (u,v)\in S_2} w_{uv} \cdot A_{uv}\\
    &> \left( \frac{1}{2} \sum_{ u, v \in i,\ (u,v) \in S_{12}} w_{uv} \cdot A_{uv}  +  \frac{1}{2} \sum_{ u, v \in i,\ (u,v) \in S_2} w_{uv} \cdot A_{uv} \right) \cdot q'_{1} \\
    &= e_{22}^{i} \cdot q'_{1}\\
	e_{2out}^{i}  &= \frac{1}{2} \sum_{u \in i, v \notin i} w_{uv} A_{uv} \\
	e'^{i}_{2out}
    &= \frac{1}{2} \sum_{u \in i, v \notin i} w_{uv} A_{uv} \cdot q'_{1}
    = e_{2out}^{i} \cdot q'_{1}
\end{align*}
Thus, $\frac{e'_{2 out}}{e_{2 out}} < \frac{e'_{2 2}}{e_{2 2}}$, and combined with Lemma 2, this proves that $Q'_2 > Q_2$, the modularity increases after ReduceWeight. 

Similarly, the modularity of layer 1 after RemoveEdge on layer 1, $Q'_1 > Q_1$.   
\end{proof}

The analysis shows that weakening one layer with any one of the methods (RemoveEdge, ReduceEdge, ReduceWeight) increases the modularity of the other layer. These results follow naturally from Lemma 2, which is in some way a stronger claim that the modularity of the remaining layer increases as long as a larger percentage of outgoing edges is reduced than internal edges.

\input{simulation.tex}

\section{Conclusion}
In this work, we provide a theoretical perspective on the hidden community detection meta-approach HICODE, on multi-layer stochastic block models. We prove that 
in synthetic two-layer stochastic blockmodel networks, the modularity of a layer will increase, after we apply a weakening method (RemoveEdge, ReduceEdge, or ReduceWeight) on all communities in the other layer, which boosts the detection of the current layer when the other layer is weakened.A simulation of relative modularity during iterations is also provided to illustrate on how HICODE weakening method works during the iterations.
Our work builds a solid theoretical base for HICODE, demonstrating that it is promising in uncovering both hidden and dominant layers of communities in two-layer stochastic block model networks. In future work, we will generalize the theoretical analysis to synthetic networks with more than two stochastic block model layers. 
\bibliographystyle{splncs04} 
\bibliography{mybib}

\newpage 
%Appendix 
\input{Append-HICODE}

\input{Append-2layerProof}

\input{Append-simulation}

\end{document}

%% file: intro.tex
\section{Introduction}
Community detection problem has occurred in a wide range of domains, from social network analysis to biological protein-protein interactions, and numerous algorithms have been proposed, based on the assumption that nodes in the same community are more likely to connect with each other. 
While many real-world social networks satisfy the assumption, their communities can overlap in interesting ways: communities based on schools can overlap as students attend different schools; connections of crime activities often hide behind innocuous social connections; proteins serving multiple functions can belong to multiple function communities. 
In any of these networks, communities can have more structures than random overlappings. For example, communities based on schools may be divided into primary school, middle school, high school, college and graduate school layers, where each layer are approximately disjoint. This observation inspires us to model real world networks as having multiple layers. 

%In the most common random graph model, Erdos-Renyi model $G(n,p)$, the graph has $n$ nodes, and each node pair has probability $p$ to form an edge. 
To simulate real-world networks, researchers also build generative models such as single-layer stochastic block model $G(n, n_1, p, q)$ ($p>q$). It can be seen as Erd\H{o}s-R\'{e}nyi model with communities---$G(n, n_1, p, 1)$ has $n$ nodes that belongs to $n_1$ disjoint blocks/communities (we use them interchangeably in the following), and any node pair internal to a community has probability $p$ to form an edge, while any node pair across two communities have $q$ probability to form an edge. 
We propose a multi-layer stochastic block model $G(n, n_1, p_1, ..., n_L, p_L)$, where each layer $l$ consists of $n_l$ disjoint communities, and communities in different layers are independent to each other. Each layer $l$ is associated with one edge probability $p_l$, determining the probability that a node pair internal to a community in that layer forms an edge. In this ideal abstraction, we assume that each node belongs to exactly one community in each layer, and an edge is generated only through that process, i.e. all edges outgoing communities of one layer are generated as internal edges in some other layers. 
Note that our model is different to the multi-layer stochastic blockmodel proposed by Paul \etal ~\cite{paul2016consistent}, where they have different types of edges, and each type of edges forms one layer of the network.

He \etal \ \cite{He15corr,he18} first introduce the concept of hidden communities, remarked as a new graph-theoretical concept~\cite{teng2016scalable}. He \etal \ propose the Hidden Community Detection (HICODE) algorithm for networks containing both strong and hidden layers of communities, where each layer consists of a set of disjoint or slightly overlapping communities. A  hidden community is a community most of whose nodes also belong to other stronger communities as measured by metrics like modularity~\cite{girvan2002community}. % or conductance~\cite{conductance2000} 
They showed through experiments that HICODE uncovers grounded communities with higher accuracy and finds hidden communities in the weak layers. However, they did not provide any theoretical support. 

In this work, we provide solid theoretical analysis that demonstrates the effectiveness of HICODE on two-layer stochastic models. 
One important step in HICODE algorithm is to reduce the strength of one partition when the partition is found to approximate one layer of communities in the network. Since communities in different layers unavoidably overlap, both internal edges and outgoing edges of remaining layers have a chance to be reduced while reducing one layer. It was unclear how the modularity of remaining layer would change. 
Through rigorous analysis of three layer weakening methods they suggested,  we prove that using any one of \textit{RemoveEdge}, \textit{ReduceEdge} and \textit{ReduceWeight} on one layer increases the modularity of the grounded partition in the unreduced layer. 
Thus, we provide evidence that HICODE's layer reduction step makes weak layers more detectable. 

In addition, through simulation, we show that on two-layer stochastic block model networks, partitions with locally maximal modularity roughly correspond to planted partitions given by grounded layers. As a result, modularity optimizing community detection algorithms such as Louvain~\cite{blondel2008fast} can approximate layers fairly accurately in a two-layer stochastic block model, even when layers are almost equally strong and non-trivially overlapped. This indicates the previous proof's assumption that one layer of communities is reduced exactly is reasonable. We also illustrate how the modularity of randomly sampled partitions change as HICODE iterates, and our plots show that not only absolute modularity but also relative modularity of unreduced layers increases as HICODE reduces one found layer.

%% file: simulation.tex
\section{Simulation of Relative Modularity}
To show whether reducing layers makes other layers more detectable when running HICODE, we simulate how grounded layers' relative modularity changes as the weakening method iterates on two-layer stochastic block models, and compare the grounded layers' modularity value with other partitions' modularity values. The number of possible partitions of $n$ nodes is exponential, so it would be computationally unrealistic just to enumerate them, let alone calculate modularity for all of them. So we employ sampling of partitions. We calculate modularity for all sampled partitions and plot them on a 2-dimensional plane based on their similarities with the grounded layer 1 and layer 2, and show the modularity values through the colormap with nearest interpolation.

\subsection{Sampling method}
We sample 2000 partitions similar to layer 1 (or 2) by starting from layer 1 (or 2), and then exchange a pair of nodes or change the membership of one node for $k = 1, ..., 500$ times. We also include 1200 partitions that mixed layer 1 and layer 2 by having $k$ randomly selected nodes getting assigned to their communities in layer 1 and the rest $200 - k $ nodes getting assigned to their communities in layer 2. As planted communities in different layers are independent, this sampling method gives a wide range of partitions while being relatively fast. 
To measure the similarity between two partitions, we adapt normalized mutual information (NMI) ~\cite{mcdaid2011normalized} for overlapping communities (The definition of NMI is in Appendix C.). 
Partitions of nodes are inherently high-dimensional. To place them on 2-dimensional plane for the plotting purpose, we use its NMI similarity with layer 1 as the $x$-coordinate, and NMI similarity with layer 2 as the $y$-coordinate.  

At each iteration, We use the modularity optimization based fast community detection algorithm~\cite{blondel2008fast} as the base algorithm to uncover a single layer of communities. 

\subsection{Simulation on ReduceEdge}
Fig. \ref{fig:twolayerblock} presents the simulated results on a two-layer block model $G(600, 15, $ $12, 0.1, 0.12)$ using ReduceEdge as the weakening method. 
In this network, layer 2 is the dominant layer (communities are bigger and denser) and layer 1 is the hidden layer. The modularity of layer 2 is 0.546, while the modularity of layer 1 is 0.398. 
We plot the modularity of the estimated layer and other sampled partitions at different iterations of HICODE. On each subfigure, the dark red cross sign denotes where the estimated layer projects on the 2-dimensional plane. Simulations using RemoveEdge and ReduceWeight yield similar results. See their plots in Appendix C. 

\begin{figure}[htbp]
	\centering
	\includegraphics[width = \textwidth]{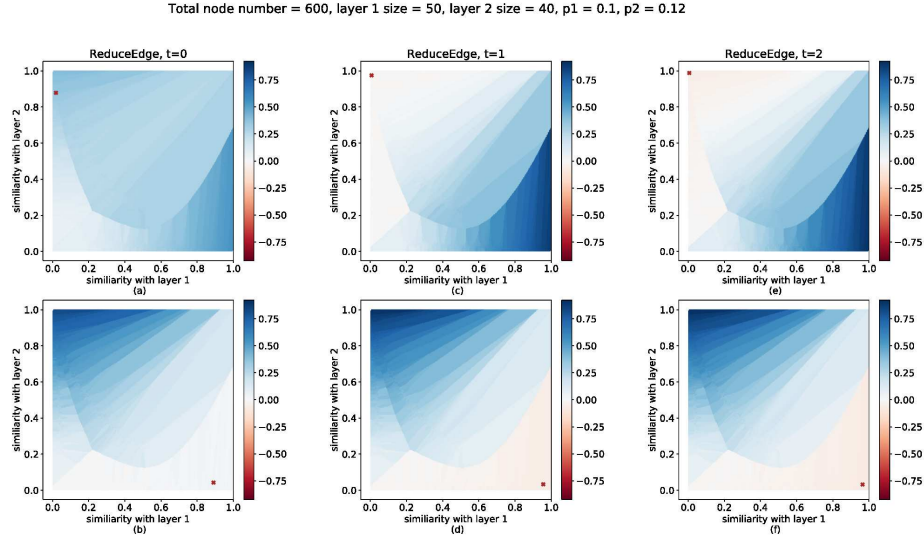}
    \vspace{-1em}
    \caption{Simulation results of ReduceEdge on $G(600, 15, 12, 0.1, 0.12)$.}
    \label{fig:twolayerblock}
\end{figure}

\vspace{-1.5em}
\begin{enumerate}
\item Initially, two grounded layers here have similar modularity values, contributing to the two local peaks of modularity, one at the right-bottom and the other at the left-top. 
\item (a): At iteration $t = 0$:, the base algorithm finds an approximate layer 2, whose NMI similarity with layer 2 is about 0.90. 
\item (b): After reducing that partition, the modularity local peak at the left-top sinks and the modularity peak at right-bottom rises, and the base algorithm finds an approximate layer 1 whose NMI similarity with layer 1 is about 0.89. ReduceEdge then reduces this approximated layer 1 and makes it easier to approximate layer 2. 
\item (c) and (d): At $t = 1$, the base algorithm finds an approximate layer 2 having 0.97 NMI similarity with layer 2, which is a significant improvement. As that more accurate approximation of layer 2 is reduced, the base algorithm is able to find a better approximation of layer 1 too. In our run, it finds an approximation that has 0.96 NMI similarity with layer 1. 
\item (e) and (f):  As HICODE iterates, at $t = 2$, the base algorithm is able to uncover an approximate layer 2 with 0.98 NMI similarity, and an approximate layer 1 with 0.97 NMI similarity.
\end{enumerate}
% \begin{table}[ht]
%   \centering
%   \begin{tabular}{lll}
% \textbf{Iterations}&\textbf{Approximate Layer 2}&\textbf{Approximate Layer 2}\\
% t=0 & (0.019, 0.878) & (0.890, 0.042) \\
% t=1 & (0.006, 0.975) &  (0.956,0.033)\\
% t=2 & (0.006,0.988) &  (0.964,0.031)\\
%   \end{tabular}
%   \caption{RemoveEdge}
% \end{table}

%% file: Append-HICODE.tex
\section*{Appendix A: Procedure of HICODE algorithm}
The HIdden COmmunity DEtection (HICODE) algorithm takes in a base algorithm $\mathcal{A}$ that finds one disjoint partition of communities~\footnote{A set of lightly overlapping communities is also allowed for the base algorithm $\mathcal{A}$. Here we only consider the partition case for simplicity.} and uses $\mathcal{A}$ to \textbf{identify} and \textbf{refine} layers of community partitions. In the identification stage, HICODE iterates the following two steps until reaching a preset number of layers: 
\begin{enumerate}
\item \textbf{Identify: } Run $\mathcal{A}$ to find one disjoint partition of communities on network $G$ and consider the partition as one layer of communities, $l$;
\item \textbf{Weaken: } Approximate edges contributed by layer $l$ on $G$ and reduce these edges on $G$.
\end{enumerate}

HICODE then refines community partitions on each layer through iterating: 
\begin{enumerate}
\item \textbf{Weaken: } Approximate edges contributed by all layers except $l$ and reduce these edges on the original network $G$;
\item \textbf{Refine: } Run $\mathcal{A}$ on the remaining network to obtain a refined community partition for layer $l$.
\end{enumerate}

%% file: Append-2layerProof.tex
\section*{Appendix B: Detailed Proofs for two-layer SBM}
\setcounter{lemma}{0}
\label{appendix2layerproof}
\begin{lemma}
In the synthetic two-layer block model network $G(n, n_{1}, n_{2}, p_{1}, p_{2})$, for any  community in layer 1, the expected number of its internal edges, its outgoing edges, and layer 1's modularity are as follows:%\footnote{We ingore $i$ in the equations, as all communities in one layer are of equal size. So $e_11$ indicates $e_11^i$, etc.}:
\begin{align}
     e_{11} &=  \left(1-\frac{1}{n_{2}}\right)m_{1}p_{1} + \frac{1}{n_{2}}m_{1}p_{12}, \label{append:eq1}\\
    e_{1out} &=  \frac{p_{2}}{n_{2}} s_{1}(n - s_{1}), \label{append:eq2}\\
    Q_1 &= 1 - \frac{1}{n_1} - \frac{e_{1out}}{d_1}, \label{append:eq3}
\end{align}
where $p_{12} = p_1 + p_2 - p_1 \cdot p_2$. Symmetrically, given a community $i$ in layer 2, the expected number of its internal edges, its outgoing edges, and layer 2's modularity are as follows: 
\begin{align}
    e_{22} &= \left(1-\frac{1}{n_{1}}\right)m_{2}p_{2} + \frac{1}{n_{1}}m_{2}p_{12}, \label{append:eq4}\\
    e_{2out} &=  \frac{p_1}{n_1} s_{2}(n - s_{2}), \label{append:eq5}\\
    Q_2 &= 1 - \frac{1}{n_2} - \frac{e_{2out}}{d_2}. \label{append:eq6}
\end{align}
\end{lemma}

\begin{proof}
  
All communities in one layer are of equal size, so for any community $i$ in a fixed layer $l$, the probability that a node belongs to $i$ is $\frac{1}{n_l}$. In addition, layers are independent, so for any pair of community $i$ in layer 1, $j$ in layer 2, the probability of a node belonging to both $i$ and $j$ is $\frac{1}{n_{1}n_{2}}$. So the expected number of nodes in the intersection of community $i$ and $j$ is $r = \frac{n}{n_{1}n_{2}}$. %, let $r = \frac{n}{n_{1}n_{2}}$. 

Denote the intersection block of community $i,j$ as $b_{ij}$. $b_{ij}$ has $r = \frac{n}{n_{1}n_{2}}$ nodes, and thus $m_{b_{ij}} = \frac{1}{2}r^2$ node pairs. For any community $i$ in layer 1, there are $n_2$ communities in layer 2 that $i$ can intersect with, and they are disjoint, so the expected number of node pairs that are internal to both $i$ and some community in layer 2 is $n_2 m_{b_{ij}}$. Since $r = \frac{n}{n_1 \cdot n_2}$, $s_1 = \frac{n}{n_1}$, $r = \frac{s_1}{n_2}$,
\begin{align*}
& m_{b_{ij}} = \frac{1}{2} r^2 = \frac{1}{n_2^2} \cdot \frac{1}{2} \cdot s_1^2 =  \frac{1}{n_2^2} m_1 \\
\implies & n_2 m_{b_{ij}} = \frac{1}{n_2} m_1.
\end{align*}

The equation indicates that for  community $i$ in layer 1, $\frac{1}{n_2} \cdot  m_1$ node pairs in layer 1 are also in the same community of layer 2. While the rest  $(1 - \frac{1}{n_2})m_1$ node pairs in $i$ form edges with probability $p_1$, those $\frac{1}{n_2} \cdot  m_1$ node pairs in the intersection form edges with probability $p_{12} = p_1 + p_2 - p_1 \cdot p_2$. Thus, the number of internal edges in any community of layer 1 is 
\[e_{11} = (1 - \frac{1}{n_2})m_1p_1 + \frac{1}{n_2}m_1p_{12}. \]
This completes the proof for Eq. \ref{append:eq1}. 

The probability that a node pair is internal in layer $2$ is $\frac{1}{n_2}$, so the number of nodes pairs outgoing from community $i$ of layer 1  that also happens to be internal in layer 2 is:  
\[ \frac{1}{n_2} \cdot
    \# \text{ of nodes pairs outgoing from $i$}  =
    \frac{1}{n_2} \cdot s_1(n -s_1). \]
Thus, the expected number of outgoing edges from community $i$ is:  
\begin{align*}
    e_{1out} &= p_2 \cdot \# \text{ of nodes pairs outgoing from $i$ that is internal to layer 2}\\
    &= \frac{p_2}{n_2} \cdot s_1(n - s_1).
\end{align*}
This completes the proof for Eq. \ref{append:eq2}. 

Also, the total number of edges, denoted as $e$, equals a half of the degree sum of all nodes, 
\begin{align*}
e = \frac{1}{2} \sum_{i \in layer\ l} d_l = \frac{1}{2} n_l \cdot d_l.
\end{align*}
Therefore, the modularity $Q^i_1$ of any community $i$ in layer 1 is 
\begin{align*}
    Q^i_1 = &\frac{e_{11}}{e} - \left(\frac{d_1}{2e} \right)^2 = \frac{2e_{11}}{n_1d_1} - \left(\frac{d_1}{n_1d_1}\right)^2
    = \frac{2e_{11}}{n_1 d_1} - \frac{1}{(n_1)^2}.
\end{align*}
Thus, the modularity of layer 1 is simply 
\begin{align*}
    Q_1 &= \sum_{i \in layer1} Q_1^i  = n_1 \cdot \left(\frac{2e_{11}}{n_1 d_1} - \frac{1}{(n_1)^2} \right) \\
    &= \frac{ 2e_{11} }{ d_1 } - \frac{1}{n_1} = 1 - \frac{1}{n_1} - \frac{e_{1out}}{d_1},
\end{align*}
where the last equation follows from $d_l = 2 e_{ll} + e_{lout}$. This completes the proof for Eq. \ref{append:eq3}. 

The proof for Eq. \ref{append:eq4}, \ref{append:eq5}, \ref{append:eq6} are analogous.
\end{proof}

\begin{lemma}
For layer $l$ in a two-layer stochastic blockmodel, if the layer weakening method (eg. RemoveEdge, ReduceEdge, ReduceWeight) reduces more percentage of outgoing edges than internal edges, i.e. the expected number of internal and outgoing edges after weakening $e'_{ll}, e'_{l out}$ satisfies $\frac{e'_{l out}}{e_{l out}} < \frac{e'_{l l}}{e_{l l}}$, then the modularity of layer $l$ increases after the weakening method.
\end{lemma}

\begin{proof}
From Lemma 1, the modularity of layer $l$ before the layer weakening is $Q_l =  1 - \frac{1}{n_l} - \frac{e_{l out}}{d_l}$, becomes $Q'_l = 1 - \frac{l}{n_l} - \frac{e'_{lout}}{d_l'}$ after weakening. The number of edges must be non-negative, so we can assume that $e_{ll}, e_{lout}, e'_{ll}$ are positive, and then
\begin{align*}
    \frac{e'_{l out}}{e_{l out}} < \frac{e'_{l l}}{e_{l l}} &\iff \frac{2 e_{l l}}{e_{l out}} + 1 < \frac{2 e'_{l l}}{e'_{l out}} + 1 \\
    &\iff \frac{e_{l out}}{2 e_{l l} + e_{l out}}  > \frac{e'_{l out}}{2 e'_{l l} + e'_{l out}}\\
    &\iff  \frac{e_{l out}}{d_l}  > \frac{e'_{l out}}{d'_l} 
    \\
    & \implies 1 - \frac{l}{n_l} - \frac{e_{lout}}{d'_l} < 
    1 - \frac{l}{n_l} - \frac{e'_{lout}}{d_l'} \\
    &\implies Q_l < Q'_l.
\end{align*}
Therefore, $\frac{e'_{l out}}{e_{l out}} < \frac{e'_{l l}}{e_{l l}}  \implies Q_l < Q'_l$.
\end{proof}

\begin{lemma}
In $G(n, n_1, n_2, p_1, p_2)$, before any weakening procedure. 
\begin{align*}
&e_{11} = \frac{|S_{12}| + |S_1|}{ n_1}, &e_{1out} = \frac{2}{n_1} |S_2|,\\
&e_{22} = \frac{|S_{12}| + |S_2|}{ n_2}, &e_{2out} = \frac{2}{n_2} |S_1|.
\end{align*}
\end{lemma}

\begin{proof}
In our two-layer stochastic block model, 
any outgoing edge of a community in layer 1 is internal to layer 2, and by definition, they are not internal to layer 1, Thus, the set of outgoing edges of communities in layer 1 is exactly the set of edges only internal to layer 2, i.e. $S_2$. There are $n_1$ communities in layer 1, each expected to have $e_{1out}$ outgoing degrees. Each edge contributes to 2 degrees, so the expected number of outgoing edges of all communities in layer 1 is $\frac{1}{2} n_1 \cdot e_{1out} $. Thus, $|S_2| = \frac{1}{2} n_1 \cdot e_{1out} $, which implies $e_{1out} = \frac{2}{n_1} |S_2|$. The proof for $e_{2out} = \frac{2}{n_2} |S_1|$ is analogous.

Any edge that is only internal to layer 1, or internal to both layer 1 and 2 is internal to exactly one community in layer 1. Thus, the set of edges in a community $i$ of layer 1 is exactly the union of $S_1$ and $S_{12}$. $S_1$ and $S_{12}$ are disjoint, so their union has size $|S_1| + |S_{12}|$. Therefore $n_1 \cdot e_{11} = |S_1| + |S_{12}|$, and $e_{11} = \frac{1}{n_1} ( |S_1| + |S_{12}|)$. The proof for $e_{22} = \frac{1}{n_2} ( |S_2| + |S_{12}|)$ is analogous. 
\end{proof}

%% file: Append-simulation.tex
\section*{Appendix C: More Simulation of Relative Modularity}

In this section, we provide the definition of NMI similarity for two partitions, and illustrate the simulation for another two weakening methods, RemoveEdge and ReduceWeight. In Fig. \ref{fig:RemoveEdge} and \ref{fig:ReduceWeight}, we see that both methods give results similar to ReduceWeight. The three weakening methods all boost the detection on dominant layer (layer 2) and hidden layer (layer 1), and converge in three iterations. 

\begin{definition}[NMI similarity]
Normalized mutual information (NMI) of two partitions $X,Y$ is defined to be
\begin{align*}
    NMI(X,Y) = \frac{2 I(X,Y)}{ H(X) + H(Y)}
\end{align*}
where $H(X)$ is the entropy of partition with $p(x)$ taken to be $|X|$
 \begin{align*}
     H(X) = - \sum_{x \in X} p(x) \log p(x)
     = - \sum_{x \in X} |x| \log |x|
 \end{align*}
 and $I(X,Y)$ measures the mutual information between $X$ and $Y$ by
\begin{align*}
    I(X,Y) &= \sum_{x \in X} \sum_{y \in Y} p(x,y) \log \frac{p(x,y)}{p(x) \cdot p(y)}\\
    &= \sum_{x \in X} \sum_{y \in Y} |x \cap y| \log \frac{ |x \cap y| }{|x| \cdot |y|}
\end{align*}
\end{definition}

\begin{figure}[htbp]
	\centering
	\includegraphics[width = \textwidth]{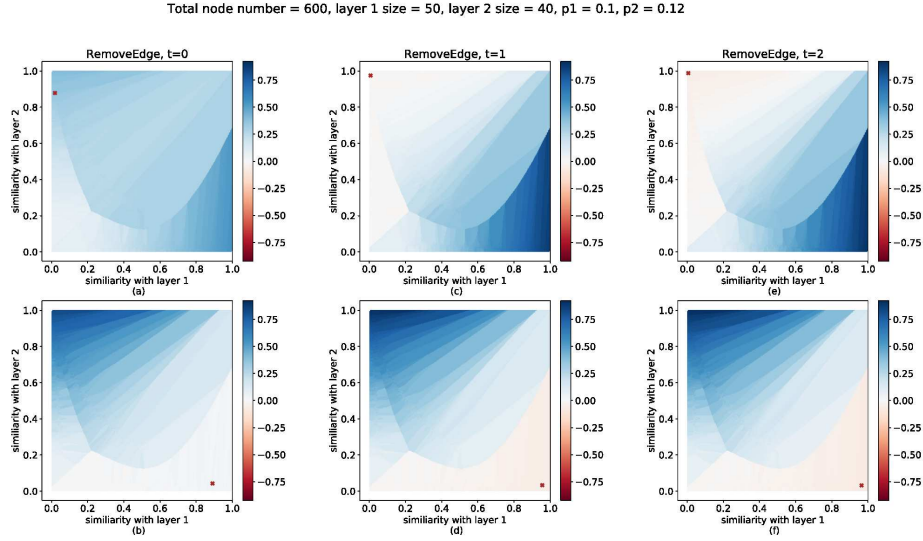}
    \caption{Simulation results of RemoveEdge on $G(600, 15, 12, 0.1, 0.12)$.}
    \label{fig:RemoveEdge}
\end{figure}

\begin{figure}[htbp]
	\centering
	\includegraphics[width = \textwidth]{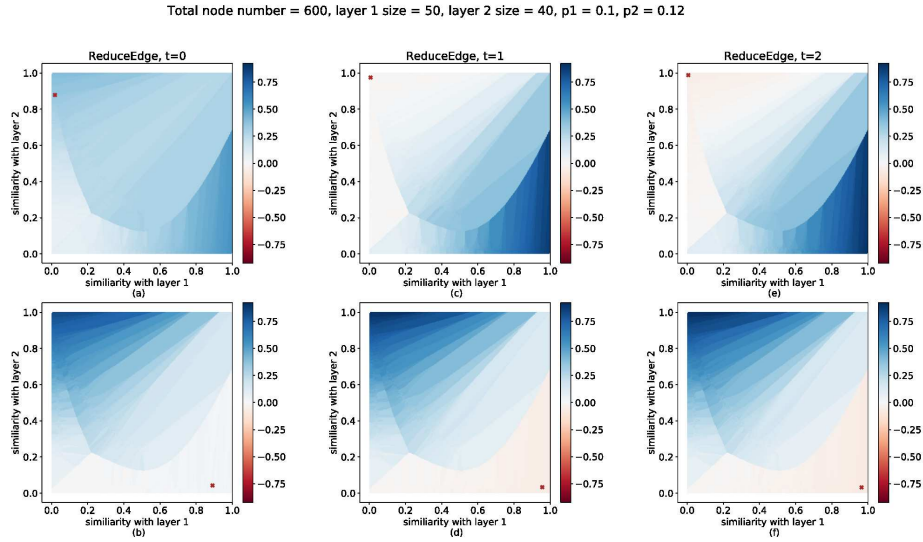}
    \caption{Simulation results of ReduceEdge on $G(600, 15, 12, 0.1, 0.12)$. The initial weight of each edge is set to 1.}
    \label{fig:ReduceWeight}
\end{figure}